\documentclass{article}
\usepackage{spconf,amsmath,graphicx,hyperref}
\usepackage{amssymb,amsfonts,amsthm}
\usepackage{algorithmic}
\usepackage{graphicx}
\usepackage{textcomp}
\usepackage{xcolor}
\usepackage[table]{xcolor} 
\usepackage{flushend}
\usepackage{orcidlink}
\usepackage{multirow}
\usepackage{lipsum}
\theoremstyle{plain}
\newtheorem{thm}{Theorem}

\newtheorem{prop}[thm]{Proposition}

\usepackage{makecell}
\usepackage{svg}
\usepackage[whole]{bxcjkjatype}
\usepackage{algorithm}
\usepackage{adjustbox}
\usepackage{colortbl}
\usepackage{cite}
\usepackage{flushend}
\usepackage{caption}

\title{Prox-friendly Log-magnitude prior on complex-valued signal}

\name{Kazuki Matsumoto$\qquad\;\;$Keidai Arai$\qquad\;\;$Kohei Yatabe\hspace{18pt}\thanks{This work was partly supported by JST FOREST Program (Grant Number JPMJFR2330, Japan).}}
\address{Tokyo University of Agriculture and Technology (TUAT), Tokyo, Japan}

\def\BibTeX{{\rm B\kern-.05em{\sc i\kern-.025em b}\kern-.08em
    T\kern-.1667em\lower.7ex\hbox{E}\kern-.125emX}}
    
\begin{document}
\ninept

\maketitle

\begin{abstract}

The logarithmic transform is essential in audio signal processing since human auditory perception is approximately logarithmic with respect to magnitude. 
However, directly incorporating prior knowledge about signals (e.g., harmonic structure) in the log-magnitude domain into optimization problems solved by standard proximal splitting algorithms remains challenging.
To address this issue, this paper proposes a novel regularizer termed EPILOG (Exponential Penalty for Imposing priors on LOG-magnitude).
EPILOG indirectly imposes prior knowledge on the log-magnitude of a complex-valued signal
through regularization of an auxiliary variable that is shown to be linked with the log-magnitude.
Furthermore, we derive its variable-wise proximity operators and develop a proximal splitting algorithm using these operators. 
Experiments on speech dereverberation demonstrate the effectiveness of the proposed regularizer, particularly in promoting cepstral-domain sparsity.

\end{abstract}

\begin{keywords}
Log-magnitude spectrogram, cepstrum, proximity operator, alternating direction method of multipliers (ADMM), speech dereverberation.
\end{keywords}

\section{Introduction}

The \emph{log-magnitude spectrogram} serves as an essential time-frequency representation in audio signal processing. It is constructed by applying an element-wise logarithm to the magnitude of a complex-valued spectrogram, which is obtained by the discrete Gabor transform (DGT) \cite{DGT,DGT2}.
The log-magnitude is widely used because human auditory perception exhibits an approximately logarithmic sensitivity to signal magnitude \cite{fastl2006psychoacoustics}.
Moreover, classical audio features, including the cepstrum \cite{1328092} and mel-frequency cepstral coefficients (MFCCs)~\cite{1163420}, are also based on log-magnitude representations, where a frequency transform such as the discrete Fourier transform (DFT) or discrete cosine transform (DCT) is applied to log-magnitude spectra.
Thus, this paper focuses on leveraging representations of the form $\mathbf{L}\log(|\mathbf{z}|)\in\mathbb{R}^K$ within a signal processing method, where $\mathbf{L}\in\mathbb{R}^{K\times N}$ is a linear operator (e.g., DCT) and $\mathbf{z}\in\mathbb{C}^N$ is a (vectorized) complex-valued spectrogram.

Optimization-based methods have been widely used in audio signal processing, where regularization plays a key role in providing a systematic way to incorporate prior knowledge of target signals into the resulting algorithms.
In particular, various structural properties of audio signals in time-frequency domain, including low-rankness~\cite{NMF_Bayesian,NMF_multichannel,ILRMA}, smoothness~\cite{OnoHPSS,SmoothSpecNMF,PhaseAwareHPSS}, and harmonic structure~\cite{harmonic_enhancement,11443227}, have been successfully utilized for many signal processing tasks.

Once such prior knowledge is introduced as a regularization term, an algorithm for solving the associated problem is required.
Proximal splitting algorithms have been widely employed for this purpose \cite{combettes2011,relax_them_all,SeparableSum,ADMM,ADMM-nonconvex}, owing to their flexibility in handling a broad class of optimization problems.
However, directly imposing prior knowledge on the log-magnitude, including sparsity of cepstrum, remains challenging within this framework, since the composition of the log-magnitude operation and a linear operator, as in $\mathbf{L}\log(|\mathbf{z}|)$, complicates the derivation of proximal splitting algorithms.

To overcome this challenge, this paper proposes
\textbf{EPILOG} (\textbf{E}xponential \textbf{P}enalty for \textbf{I}mposing priors on \textbf{LOG}-magnitude),
a regularization technique for incorporating various prior knowledge in the log-magnitude domain.
Instead of directly regularizing the log-magnitude of a complex-valued signal, EPILOG regularizes an auxiliary variable that is linked with the log-magnitude, extending a recently developed perspective-function-based method tailored to magnitude regularization \cite{LOP,araiVersatileTimeFrequencyRepresentations2023,10681150,kurodaConvexNonconvexFrameworkEnhancing2024}.
Furthermore, we provide an example of promoting cepstral-domain sparsity using EPILOG by imposing sparsity on the DCT coefficients of the log-magnitude representation.
In addition, we derive the variable-wise proximity operators for the proposed regularizer and develop an alternating direction method of multipliers (ADMM)--based algorithm.

The main contributions of this paper are summarized as follows:
\textbf{(1)} we formulate the EPILOG regularizer to enable the incorporation of various prior knowledge in the log-magnitude domain;
\textbf{(2)} we derive the variable-wise proximity operators for the EPILOG regularizer and develop an ADMM-based algorithm using these operators; and
\textbf{(3)} we demonstrate an application of the proposed method in speech dereverberation, experimentally confirming the effectiveness of the EPILOG regularizer and cepstral-domain sparsity.

\textbf{Notation}:
$\mathbb{R}$, $\mathbb{R}_+$, and $\mathbb{C}$ denote the sets of all real numbers, nonnegative real numbers, and complex numbers, respectively.
Bold lowercase and uppercase letters represent vectors and matrices, e.g.,
$\mathbf{x}=(x_n)_{n=1}^N\in\mathbb{R}^N$ and
$\mathbf{A}\in\mathbb{R}^{I\times J}$, respectively.
Operations $(\cdot)^2$, $|\cdot|$, $\sqrt{\cdot}$, and $\log(\cdot)$ are applied element-wise.
The proximity operator for a function $F:\mathbb{R}^N\to\mathbb{R}\cup\{+\infty\}$ and a parameter $\mu>0$ is defined by $\operatorname{prox}_{\mu F}(\mathbf{u})\in\arg\min_{\mathbf{v}\in\mathbb{R}^N}(F(\mathbf{v})+\frac{1}{2\mu}\|\mathbf{v}-\mathbf{u}\|_2^2)$ \cite{bauschkeConvexAnalysisMonotone2017}.

\section{Regularization of Magnitude Spectrogram}

Optimization-based methods are effective in a wide range of signal processing tasks.
For example, an audio signal recovery task, which aims to recover a signal $\mathbf{s}\in\mathbb{R}^L$ from a noisy and/or degraded observation $\mathbf{y}\in\mathbb{R}^L$, is typically formulated as
\begin{equation}
\min_{
\mathbf{x}\in\mathbb{R}^{L},\,\mathbf{z}\in\mathbb{C}^{N}
}
\left(
F_{\mathbf{y}}(\mathbf{x})
+\lambda\Omega(\mathbf{z})
\right)
\quad
\mathrm{s.t.}\quad
\mathbf{z}=\mathbf{G}\mathbf{x},
\label{eq:signal_recovery}
\vspace{-4pt}
\end{equation}
where $\mathbf{x}\in\mathbb{R}^L$ is the estimate of the signal $\mathbf{s}$,
$\mathbf{z}\in\mathbb{C}^N$ is the (vectorized) complex-valued spectrogram of $\mathbf{x}$,
$F_{\mathbf{y}}:\mathbb{R}^L\to\mathbb{R}\cup\{+\infty\}$ is a data-fidelity function ensuring the consistency with the observed signal $\mathbf{y}$, 
$\Omega:\mathbb{C}^N\to\mathbb{R}\cup\{+\infty\}$ is a regularization function inducing some prior knowledge about the signal to be recovered, 
$\lambda>0$ is a regularization parameter, 
and $\mathbf{G}:\mathbb{R}^L\to\mathbb{C}^{N}$ is a DGT operator.

Many practical audio applications assume that the magnitude $|\mathbf{z}|\in\mathbb{R}_+^N$ of the complex-valued spectrogram has a certain structure, such as low-rankness\cite{NMF_Bayesian,NMF_multichannel,ILRMA}, smoothness\cite{OnoHPSS, SmoothSpecNMF, PhaseAwareHPSS}, and harmonic structure\cite{harmonic_enhancement,11443227}.
Promoting these structures often requires a regularization function
$\Omega(\mathbf{z})=R\left(\mathbf{L}|\mathbf{z}|\right)$,
where the choice of a linear operator $\mathbf{L}\in\mathbb{R}^{K\times N}$ and a function $R:\mathbb{R}^K\to\mathbb{R}\cup\{+\infty\}$ determines the promoted structure (e.g., composition of a difference operator and a squared $\ell_2$-norm promotes smoothness).
However, 
such formulation obstructs the direct application of proximal splitting algorithms that handle $R$ by using its proximity operator, since the composition of the nonlinear operator $|\cdot|$ and the linear operator $\mathbf{L}$ makes it difficult to derive the algorithms.

To address this issue, a framework based on the properties of \emph{perspective functions} \cite{perspective, perspective2, perspective3} has recently been developed to enable flexible magnitude regularization \cite{LOP,araiVersatileTimeFrequencyRepresentations2023,10681150,kurodaConvexNonconvexFrameworkEnhancing2024}.
This framework introduces a function involving an auxiliary variable $\boldsymbol{\sigma}\in\mathbb{R}_+^N$, given by
\vspace{-1.5pt}
\begin{equation}
    \Omega_{\mathrm{Perspective}}^{(\mathbf{L},R,\gamma,\alpha)}(\mathbf{z})
    =
    \min_{\boldsymbol{\sigma}\in\mathbb{R}_+^N}
    \Bigg(
    \underbrace{    \sum_{n=1}^{N}
    \phi^{(\alpha)}(z_n,\sigma_n)}_{=\Phi^{(\alpha)}(\mathbf{z},\boldsymbol{\sigma})}
    +
    \gamma R(\mathbf{L}\boldsymbol{\sigma})
    \Bigg),
    \label{eq:perspective_regularization}
    \vspace{-1.5pt}
\end{equation}
where $\gamma>0$, $\alpha>0$, and $\phi^{(\alpha)}:\mathbb{C}\times\mathbb{R}_+\to\mathbb{R}\cup\{+\infty\}$ is the perspective function of $|\cdot|^2/2+\alpha^2/2$ given as follows \cite{perspective3}:
\vspace{-1.5pt}
\begin{align}
    \phi^{(\alpha)}(z,\sigma)
    &=\left\{
    \begin{array}{cl}
        \frac{|z|^2}{2\sigma}
        +
        \frac{\alpha^2\sigma}{2},
        & (\sigma>0),\\
        0
        & (z=0 \text{ and }\sigma=0),\\
        +\infty,
        & (\text{otherwise}).
    \end{array}
    \right.
    \label{eq:perspective-elm}
    \vspace{-1.5pt}
\end{align}

The function $\Omega_{\mathrm{Perspective}}^{(\mathbf{L},R,\gamma,\alpha)}(\mathbf{z})$ in Eq.~\eqref{eq:perspective_regularization} realizes indirect regularization of the magnitude $|\mathbf{z}|$ through the term $\gamma R(\mathbf{L}\boldsymbol{\sigma})$. 
The key mechanism is the coupling between $|\mathbf{z}|$ and the auxiliary variable $\boldsymbol{\sigma}$ induced by $\Phi^{(\alpha)}(\mathbf{z},\boldsymbol{\sigma})$, as supported in the following property \cite{araiVersatileTimeFrequencyRepresentations2023}.

\begin{prop}[Magnitude link]
\label{prop:perspective}
For any fixed $\mathbf{z}_0\in\mathbb{C}^N$, the minimizer of $\Phi^{(\alpha)}(\mathbf{z}_0,\cdot)$ is given by $\boldsymbol{\sigma}^\star_{\mathbf{z}_0}=|\mathbf{z}_0|/{\alpha}\in\mathbb{R}_+^N$.
\end{prop}

\noindent Prop.~\ref{prop:perspective} shows that minimizing $\Phi^{(\alpha)}(\mathbf{z},\boldsymbol{\sigma})$ encourages $\boldsymbol{\sigma}$ to match the magnitude $|\mathbf{z}|$ up to the scaling factor $1/\alpha$.
Accordingly, the structure imposed on $\boldsymbol{\sigma}$ by the term $\gamma  R(\mathbf{L}\boldsymbol{\sigma})$ is reflected in $|\mathbf{z}|$ through the weighted quadratic term $|z_n|^2/(2\sigma_n)$.

The term $\Phi^{(\alpha)}(\mathbf{z},\boldsymbol{\sigma})$ is compatible with proximal splitting algorithms owing to its analytically computable proximity operator \cite{LOP,perspective3}, and has been incorporated into structured time-frequency analysis \cite{araiVersatileTimeFrequencyRepresentations2023} and harmonic/percussive source separation \cite{10681150}.

\section{Proposed Method}

We propose \textbf{EPILOG}, a technique designed for flexible regularization in the log-magnitude domain. Our core methodology is to modify the magnitude-link property of the perspective function described in Prop.~\ref{prop:perspective}, enabling an auxiliary variable to be linked with the \emph{log-magnitude} of a complex-valued variable.

\subsection{EPILOG regularizer}
\label{sec:EPILOG_regularization}

We introduce a bivariate function $\psi^{(\alpha,\epsilon)}:\mathbb{C}\times\mathbb{R}\to\mathbb{R}$ defined as
\begin{equation}
\psi^{(\alpha,\epsilon)}(z,\sigma)
=
\frac{|z|^2+\epsilon}{2\exp(2\sigma)}
+
\alpha^2\sigma,
\label{eq:EPILOG-elm}
\end{equation}
where $\alpha>0$ and $\epsilon>0$.
As visualized in Fig.~\ref{fig:EPILOG_aligning}, for fixed $\sigma$, this function is quadratic with respect to $|z|$ (right panel), whereas for fixed $z$, it consists of an exponential term and a linear term with respect to $\sigma$ (middle panel).
Note that $\sigma$ is not constrained to be nonnegative unlike $\varphi^{(\alpha)}$ in Eq.~\eqref{eq:perspective-elm}, hence $\epsilon>0$ is required to keep $\psi^{(\alpha,\epsilon)}$ bounded below; if $\epsilon=0$, $\lim_{\sigma\to-\infty}\psi^{(\alpha,0)}(0,\sigma)=-\infty$.

Using Eq.~\eqref{eq:EPILOG-elm}, we define the \textbf{EPILOG regularizer} by
\vspace{-1pt}
\begin{equation}
    \Omega_{\mathrm{EPILOG}}^{(\mathbf{L},R,\gamma,\alpha,\epsilon)}
    (\mathbf{z})
    =
    \min_{\boldsymbol{\sigma}\in\mathbb{R}^N}
    \Bigg(
    \underbrace{
    \sum_{n=1}^{N}
    \psi^{(\alpha,\epsilon)}(z_n,\sigma_n)
    }_{
    =\Psi^{(\alpha,\epsilon)}
    (\mathbf{z},\boldsymbol{\sigma})
    }
    +
    \gamma R(\mathbf{L}\boldsymbol{\sigma})
    \Bigg),
    \label{eq:EPILOG_regularization}
    \vspace{-2pt}
\end{equation}
where $\mathbf{L}\in\mathbb{R}^{K\times N}$,
$R:\mathbb{R}^{K}\to\mathbb{R}\cup\{+\infty\}$,
and $\gamma>0$. 
EPILOG realizes indirect regularization of the log-magnitude $\log(|\mathbf{z}|)$, where $\gamma R(\mathbf{L}\boldsymbol{\sigma})$ imposes structure on the auxiliary variable $\boldsymbol{\sigma}$ and $\Psi^{(\alpha,\epsilon)}
    (\mathbf{z},\boldsymbol{\sigma})$ transfers it to $\log(|\mathbf{z}|)$.
This structural transfer is enabled by the following property of $\psi^{(\alpha,\epsilon)}$, which is analogous to the magnitude-link property of the perspective function $\phi^{(\alpha)}$ in Prop.~\ref{prop:perspective}.

\begin{prop}[Log-magnitude link]
\label{prop:EPILOG_logmag}
For any fixed $\mathbf{z}_0\in\mathbb{C}^N$, the minimizer of
$\Psi^{(\alpha,\epsilon)}(\mathbf{z}_0,\cdot)$, denoted by $\boldsymbol{\sigma}_{\mathbf{z}_0}^\star\in\mathbb{R}^N$, is given by
\begin{equation}
    \boldsymbol{\sigma}_{\mathbf{z}_0}^{\star}
    =
    \log\left(\!\sqrt{|\mathbf{z}_0|^2+\epsilon}\right)
    -
    \log(\alpha)
    \label{eq:EPILOG_minimizer}
\end{equation}
\end{prop}

\begin{proof}[Proof sketch]
The result follows from the optimality condition.
\end{proof}

\noindent
Prop.~\ref{prop:EPILOG_logmag} shows that minimizing the term
$\Psi^{(\alpha,\epsilon)}(\mathbf{z},\boldsymbol{\sigma})$
encourages $\boldsymbol{\sigma}$ to match
$\log(\sqrt{|\mathbf{z}|^2+\epsilon})\approx \log(|\mathbf{z}|)$ for sufficiently small $\epsilon$, up to the offset $-\log(\alpha)$, coupling $\mathbf{z}$ and $\boldsymbol{\sigma}$ in the log-magnitude domain.
The trajectory of $(z_0,\sigma_{z_0}^{\star})\in\mathbb{R}^2$ obtained by varying $z_0\in\mathbb{R}$ (i.e., the red curve in the left panel of Fig.~\ref{fig:EPILOG_aligning}) illustrates this property, where this curve approximates $\sigma=\log(|z|)$.

\begin{figure}
    \centering
    \includegraphics[scale=1]{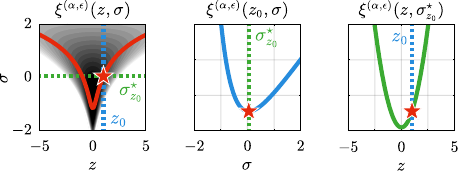}
    \vspace{-5pt}
    \caption{
    Visualization of $\psi^{(\alpha,\epsilon)}(z,\sigma)$ in Eq.~\eqref{eq:EPILOG-elm} for $(\alpha,\epsilon)=(1,0.1)$, where $z\in\mathbb{R}$ here.
    The left panel shows the values of $\psi^{(\alpha,\epsilon)}(z,\sigma)$, where brighter shades represent larger values.
    The middle and right panels show its slices at $z=z_0=1$ (blue dotted line) and $\sigma=\sigma_{z_0}^\star=\log(\sqrt{1+\epsilon})$ (green dotted line).
    The red star in the middle panel shows $\sigma_{z_0}^\star=\arg \min_{\sigma\in\mathbb{R}}\psi^{(\alpha,\epsilon)}(z_0,\sigma)$.
    The red curve in the left panel is the trajectory of $(z_0,\sigma_{{z}_0}^\star)$ obtained by varying $z_0\in\mathbb{R}$, illustrating the log-magnitude-link property in Prop.~\ref{prop:EPILOG_logmag}. 
    }
    \label{fig:EPILOG_aligning}
    \vspace{-3pt}
\end{figure}

\subsection{Example: Cepstral-domain sparsity via EPILOG}
\label{sec:EPILOG_cepstrum}

As a use case of the EPILOG regularizer, we propose \emph{cepstral-domain regularization}. 
Cepstral analysis \cite{1328092,1163420} applies a frequency transformation $\mathbf{L}$ (such as DCT or DFT along the frequency direction) to a log-magnitude spectrogram, expressed as $\mathbf{L}\log(|\mathbf{z}|)$. 
This operation decomposes the spectrum into a global spectral envelope and fine periodic structures. 
Since the cepstrum of a clean speech signal is typically sparse, we exploit this sparsity as prior knowledge.

By using the EPILOG regularizer $\Omega_{\mathrm{EPILOG}}^{(\mathbf{L},R,\gamma,\alpha,\epsilon)}(\mathbf{z})$ in Eq.~\eqref{eq:EPILOG_regularization}, the cepstral-domain sparsity can be imposed by setting $\mathbf{L}$ as a DCT operator along the frequency direction and $R$ as a sparsity-inducing function. Since $\boldsymbol{\sigma}$ is linked with the log-magnitude $\log(|\mathbf{z}|)$, the term $\gamma R(\mathbf{L}\boldsymbol{\sigma})$ indirectly regularizes the cepstrum $\mathbf{L}\log(|\mathbf{z}|)$.

Interestingly, the dependence of the EPILOG regularizer on the global signal scale can be reduced by leaving the direct current (DC) component unregularized.
Indeed, for a globally scaled spectrogram $s\mathbf{z}$ with the scaling factor $s>0$, the log-magnitude-link property in Eq.~\eqref{eq:EPILOG_minimizer} gives
$\boldsymbol{\sigma}^{\star}_{s\mathbf{z}}
=
\log(\!\sqrt{s^2|\mathbf{z}|^2+\epsilon})
-
\log(\alpha)
\approx
\log(|\mathbf{z}|)
+
\log(s/\alpha)$
for sufficiently small $\epsilon$.
Thus, the global \mbox{scale $s$}, as well as the parameter $\alpha$, appears only as an additive constant in $\boldsymbol{\sigma}^{\star}_{s\mathbf{z}}$, which is mapped by the DCT to the DC coefficient.
Therefore, excluding the DC coefficient from regularization reduces the scale dependence.
To this end, we employ the weighted $\ell_1$-norm given by
\begin{equation}
    \|\mathbf{u}\|_{\mathbf{w}}
    =
    \|\mathbf{w}\odot\mathbf{u}\|_1,
    \label{eq:weighted}
\end{equation}
where $\odot$ denotes element-wise multiplication and $\mathbf{w}\in\mathbb{R}_+^N$ is a weight vector whose entries corresponding to the DC coefficients are set to $0$, while the others are set to $1$.

\subsection{ADMM-based algorithm for EPILOG}
\label{sec:EPILOG_admm}

This section aims to construct an optimization algorithm tailored to the EPILOG regularizer.
The function $\psi^{(\alpha,\epsilon)}(z,\sigma)$ is not jointly convex with respect to $(z,\sigma)$, and hence uniqueness of its joint proximity operator is not guaranteed.
Nevertheless, it is convex with respect to each variable individually.
Based on this property, we derive the variable-wise proximity operators of $\psi^{(\alpha,\epsilon)}$ as an ingredient for constructing optimization algorithms for EPILOG regularization.

\begin{prop}
\label{thm:EPILOG_prox}
For $\mu>0$, the proximity operators of
$\psi^{(\alpha,\epsilon)}$ in Eq.~\eqref{eq:EPILOG-elm} with respect to $z$ and $\sigma$ are respectively given by
\begin{align}
    \!\!\!\!\operatorname{prox}_{\mu\psi^{(\alpha,\epsilon)}(\cdot,\sigma)}(x)
    &=
    \frac
    {x}{1+\mu\exp(-2\sigma)},
    \label{eq:EPILOUGE_coordinate_prox1}
    \\
    \!\!\!\!\operatorname{prox}_{\mu\psi^{(\alpha,\epsilon)}(z,\cdot)}(\sigma)
    &=
    -\frac{1}{2}
    \operatorname{prox}_{2\mu(|z|^2+\epsilon)\exp}
    (2(\mu\alpha^2-\sigma)),
    \label{eq:EPILOUGE_coordinate_prox2}
\end{align}
where
$\operatorname{prox}_{\widetilde{\mu}\exp}(u)
=u-W_0(\widetilde{\mu}\exp(u))$ and $W_0$ is the Lambert W-function, i.e., the inverse of $\xi\mapsto \xi\exp(\xi)$ on $\xi\in[-1,+\infty)$.
\end{prop}

\begin{proof}[Proof sketch]
The result follows from the basic properties of proximity operators \cite{SeparableSum} and the proximity operator of $\exp(\cdot)$ \cite{bauschkeConvexAnalysisMonotone2017}.
\end{proof}

Using these proximity operators, we propose an ADMM-based algorithm by alternately minimizing the augmented Lagrangian with respect to the primal variables.
To simplify the algorithm, we assume that the DGT operator $\mathbf{G}$ is Parseval tight \cite{janssenCharacterizationComputationCanonical2002}, i.e., $\mathbf{G}^{\mathsf H}\mathbf{G}=\mathbf{I}$, where $\mathbf{I}$ is the identity matrix.
We substitute the EPILOG regularizer in Eq.~\eqref{eq:EPILOG_regularization} into Eq.~\eqref{eq:signal_recovery} as $\Omega = \Omega_{\mathrm{EPILOG}}^{(\mathbf{L},R,\gamma,\alpha,\epsilon)}$. 
Then, we apply variable splitting to obtain the following optimization problem:
\begin{equation}
\begin{array}{cl}
    \displaystyle
    \min_{\substack{
        \mathbf{x}\in\mathbb{R}^{L}
        \mathbf{z}\in\mathbb{C}^{N},\\
        \boldsymbol{\sigma}\in\mathbb{R}^{N},\,
        \widetilde{\boldsymbol{\sigma}}\in\mathbb{R}^{N},\,
        \boldsymbol{\tau}\in\mathbb{R}^{K}
    }}
    &\!\!\!\!\!\!\!\!\!
    \left(
    F_{\mathbf{y}}(\mathbf{x})
    +
    \lambda\Psi^{(\alpha,\epsilon)}
    (\mathbf{z},\boldsymbol{\sigma})
    +
    \lambda\gamma R(\boldsymbol{\tau})
    \right)
    \\
    \mathrm{s.t.}
    &
    \mathbf{z}=\mathbf{G}\mathbf{x},
    \quad\boldsymbol{\sigma}=\widetilde{\boldsymbol{\sigma}},
    \quad\boldsymbol{\tau}=\mathbf{L}\widetilde{\boldsymbol{\sigma}}.
    \end{array}
    \label{eq:EPILOG_split}
\end{equation}
The augmented Lagrangian \cite{ADMM} associated with Eq.~\eqref{eq:EPILOG_split} is given by
\begin{align}
\begin{array}{@{}r@{\;}l@{\;}l@{\;}l@{}}
    \displaystyle \mathcal{L}_{\rho}
    &\displaystyle \!{}={} F_{\mathbf{y}}(\mathbf{x})
    &\displaystyle \!{}+{}
    \left\langle
    \boldsymbol{\zeta}_{\mathbf{z}},
    \mathbf{G}\mathbf{x}-\mathbf{z}
    \right\rangle
    &\displaystyle \!{}+{}
    \frac{\rho}{2}
    \left\|
    \mathbf{G}\mathbf{x}-\mathbf{z}
    \right\|_2^2
    \\[8pt]
    &\displaystyle \!{}+{}
    \lambda\Psi^{(\alpha,\epsilon)}
    (\mathbf{z},\boldsymbol{\sigma})
    &\displaystyle \!{}+{}
    \left\langle
    \boldsymbol{\zeta}_{\boldsymbol{\sigma}},
    \widetilde{\boldsymbol{\sigma}}-\boldsymbol{\sigma}
    \right\rangle
    &\displaystyle \!{}+{}
    \frac{\rho}{2}
    \left\|
    \widetilde{\boldsymbol{\sigma}}-\boldsymbol{\sigma}
    \right\|_2^2
    \label{eq:EPILOG_augmented_lagrangian}
    \\[8pt]
    &\displaystyle \!{}+{}
    \lambda\gamma R(\boldsymbol{\tau})
    &\displaystyle \!{}+{}
    \left\langle
    \boldsymbol{\zeta}_{\boldsymbol{\tau}},
    \mathbf{L}\widetilde{\boldsymbol{\sigma}}-\boldsymbol{\tau}
    \right\rangle
    &\displaystyle \!{}+{}
    \frac{\rho}{2}
    \left\|
    \mathbf{L}\widetilde{\boldsymbol{\sigma}}-\boldsymbol{\tau}
    \right\|_2^2.
\end{array}
\end{align}
where
$\boldsymbol{\zeta}_{\mathbf{z}}\in\mathbb{C}^{N}$,
$\boldsymbol{\zeta}_{\boldsymbol{\sigma}}\in\mathbb{R}^{N}$,
and
$\boldsymbol{\zeta}_{\boldsymbol{\tau}}\in\mathbb{R}^{K}$
are unscaled dual variables associated with the constraints in Eq.~\eqref{eq:EPILOG_split}, $\rho>0$ is a parameter, and
$\langle\cdot,\cdot\rangle$ denotes the real inner product.

\begin{algorithm}[t]
\caption{Audio signal recovery with the EPILOG regularizer}
\label{alg:EPILOG}
\begin{algorithmic}[1]
\vspace{-2pt}
\STATE Initialize
$\mathbf{x}^{[0]}$,
$\mathbf{z}^{[0]}$,
$\boldsymbol{\sigma}^{[0]}$,
$\widetilde{\boldsymbol{\sigma}}^{[0]}$,
$\boldsymbol{\tau}^{[0]}$,
$\widetilde{\boldsymbol{\zeta}}_{\mathbf{z}}^{[0]}$,
$\widetilde{\boldsymbol{\zeta}}_{\boldsymbol{\sigma}}^{[0]}$,
and
$\widetilde{\boldsymbol{\zeta}}_{\boldsymbol{\tau}}^{[0]}$.

\FOR{$k=0,1,2,\ldots$}

\STATE \hspace{1pt}$\displaystyle
\mathbf{x}^{[k+1]}
=
\operatorname{prox}_{F_{\mathbf{y}}/\rho}
(
\mathbf{G}^{\mathsf H}
(
\mathbf{z}^{[k]}-\widetilde{\boldsymbol{\zeta}}_{\mathbf{z}}^{[k]}
)
)$

\STATE $\displaystyle
\widetilde{\boldsymbol{\sigma}}^{[k+1]}
=
(
\mathbf{I}+\mathbf{L}^{\mathsf T}\mathbf{L}
)^{-1}
(
\boldsymbol{\sigma}^{[k]}
-
\widetilde{\boldsymbol{\zeta}}_{\boldsymbol{\sigma}}^{[k]}
+
\mathbf{L}^{\mathsf T}
(
\boldsymbol{\tau}^{[k]}-\widetilde{\boldsymbol{\zeta}}_{\boldsymbol{\tau}}^{[k]}
)
)$

\STATE \hspace{1.9pt}$\displaystyle
\mathbf{z}^{[k+1]}
=
\operatorname{prox}_{
(\lambda/\rho)
\Psi^{(\alpha,\epsilon)}
(\cdot,\boldsymbol{\sigma}^{[k]})
}
(
\mathbf{G}\mathbf{x}^{[k+1]}
+
\widetilde{\boldsymbol{\zeta}}_{\mathbf{z}}^{[k]}
)$ \hspace{3pt}$ \triangleright$ Eq.~\eqref{eq:EPILOUGE_coordinate_prox1}

\STATE $\displaystyle
\boldsymbol{\sigma}^{[k+1]}
=
\operatorname{prox}_{
(\lambda/\rho)
\Psi^{(\alpha,\epsilon)}
(\mathbf{z}^{[k+1]},\cdot)
}
(
\widetilde{\boldsymbol{\sigma}}^{[k+1]}
+
\widetilde{\boldsymbol{\zeta}}_{\boldsymbol{\sigma}}^{[k]}
)$ \hspace{3pt}$\triangleright$ Eq.~\eqref{eq:EPILOUGE_coordinate_prox2}

\STATE \hspace{0.75pt}$\displaystyle
\boldsymbol{\tau}^{[k+1]}
=
\operatorname{prox}_{(\lambda\gamma/\rho)R}
(
\mathbf{L}\widetilde{\boldsymbol{\sigma}}^{[k+1]}
+
\widetilde{\boldsymbol{\zeta}}_{\boldsymbol{\tau}}^{[k]}
)$

\STATE \hspace{1.3pt}$\displaystyle
\widetilde{\boldsymbol{\zeta}}_{\mathbf{z}}^{[k+1]}
=
\widetilde{\boldsymbol{\zeta}}_{\mathbf{z}}^{[k]}
+
\mathbf{G}\mathbf{x}^{[k+1]}
-
\mathbf{z}^{[k+1]}$

\STATE \hspace{1.3pt}$\displaystyle
\widetilde{\boldsymbol{\zeta}}_{\boldsymbol{\sigma}}^{[k+1]}
=
\widetilde{\boldsymbol{\zeta}}_{\boldsymbol{\sigma}}^{[k]}
+
\widetilde{\boldsymbol{\sigma}}^{[k+1]}
-
\boldsymbol{\sigma}^{[k+1]}$

\STATE \hspace{1.3pt}$\displaystyle
\widetilde{\boldsymbol{\zeta}}_{\boldsymbol{\tau}}^{[k+1]}
=
\widetilde{\boldsymbol{\zeta}}_{\boldsymbol{\tau}}^{[k]}
+
\mathbf{L}\widetilde{\boldsymbol{\sigma}}^{[k+1]}
-
\boldsymbol{\tau}^{[k+1]}$

\ENDFOR
\end{algorithmic}
\end{algorithm}
\setlength{\textfloatsep}{20.8pt}

The algorithm for solving Eq.~\eqref{eq:EPILOG_split} is summarized in Algorithm~\ref{alg:EPILOG}, where $\mathcal{L}_\rho$ is alternately minimized with respect to the primal variables $\mathbf{x}$, $\widetilde{\boldsymbol{\sigma}}$, $\mathbf{z}$, $\boldsymbol{\sigma}$, and $\boldsymbol{\tau}$, while the scaled dual variables $(\widetilde{\boldsymbol{\zeta}}_{\mathbf{z}},\widetilde{\boldsymbol{\zeta}}_{\boldsymbol{\sigma}},\widetilde{\boldsymbol{\zeta}}_{\boldsymbol{\tau}})=(\boldsymbol{\zeta}_{\mathbf{z}},\boldsymbol{\zeta}_{\boldsymbol{\sigma}},\boldsymbol{\zeta}_{\boldsymbol{\tau}})/\rho$ are updated by the corresponding gradient-ascent steps.
The proximity operators for $\Psi^{(\alpha,\epsilon)}$ with respect to $\mathbf{z}$ and $\boldsymbol{\sigma}$ are evaluated element-wise using the proximity operators of $\psi^{(\alpha,\epsilon)}$ derived in Prop.~\ref{thm:EPILOG_prox}.
The convergence analysis of this algorithm is left as future work.

\begin{figure*}
    \centering
    \includegraphics[scale=1]{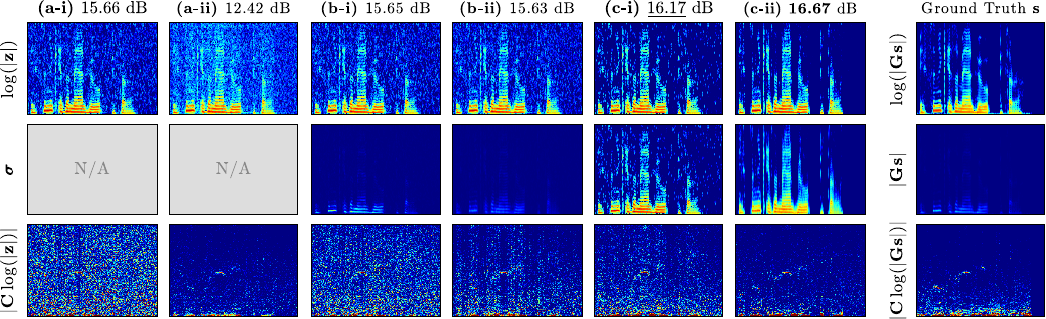}
    \vspace{-4pt}
\caption{
Examples of the variables obtained by each method.
The left six columns show the results of \textbf{(a-i)}--\textbf{(c-ii)}, where $\textbf{(c-ii)}$ is the proposed method that promotes \emph{cepstral-domain sparsity}. 
The above numbers indicate the SI-SNR scores in dB, where bold and underlined numbers represent the best and second-best scores, respectively.
The top and middle rows show the log-magnitude spectrograms $\log(|\mathbf z|)$ of the estimated signals $\mathbf{z}$ and its corresponding auxiliary variables $\boldsymbol{\sigma}$, respectively; methods \textbf{(a)} have no auxiliary variables and hence ``N/A'' is shown.
The bottom row shows the magnitudes of the DCT coefficients (i.e., the magnitudes of the cepstral coefficients) \(|\mathbf{C}\log(|\mathbf z|)|\), where $\mathbf{C}$ denotes the DCT operator along the frequency direction.
The rightmost column shows the ground-truth signal in the log-magnitude domain $\log(|\mathbf{G}\mathbf{s}|)$ (top), in the magnitude domain $|\mathbf{G}\mathbf{s}|$ (middle), and in the cepstral domain $|\mathbf{C}\log(|\mathbf{G}\mathbf{s}|)|$ (bottom), where $\mathbf{G}$ is the DGT operator.
}
\vspace{-4pt}
\label{fig:placeholder}
\end{figure*}

\section{Application to Speech Dereverberation}
\label{sec:experiment}

We conducted an experiment on speech dereverberation to confirm the effectiveness of the EPILOG regularizer.
The observed signal was generated by $\mathbf{y}=\mathbf{A}\mathbf{s}+\mathbf{n}$, where $\mathbf{s}\in\mathbb{R}^{L}$ is a clean speech signal, $\mathbf{A}\in\mathbb{R}^{L\times L}$ is a convolution matrix corresponding to a room impulse response (RIR), and $\mathbf{n}\in\mathbb{R}^L$ is additive Gaussian noise.
Clean speech signals were taken from the LibriTTS-R test-clean dataset \cite{koizumi_libritts-r_2023}, and RIRs were taken from the BUT Reverb Database \cite{szoke_building_2019}.
All the signals used in the experiment were resampled to 8~kHz.
For the DGT, we used a Parseval-tight Hann window with a length of 256 samples and a hop size of 128 samples.
The signals were truncated to 128 time frames in the DGT domain.
The noise level was set to $-20$~dB relative to $\mathbf{A}\mathbf{s}$.
The signal $\mathbf{s}$ and the matrix $\mathbf{A}$ were normalized so that both $\mathbf{s}$ and $\mathbf{A}\mathbf{s}$ have the root mean square of 1.

\subsection{Algorithms and parameters}
\label{sec:compared_methods}

We compared six methods to solve the dereverberation problem formulated in Eq.~\eqref{eq:signal_recovery}, where the data-fidelity function is set to $F_{\mathbf{y}}(\mathbf{x}) = (1/2)\|\mathbf{A}\mathbf{x}-\mathbf{y}\|_2^2$.
The methods are categorized into groups \textbf{(a)}, \textbf{(b)}, and \textbf{(c)} based on their regularization domains:
\vspace{2pt}
\begin{itemize}
    \item[{\makebox[1.1em][c]{\textbf{(a)}}}]
    complex-valued: $\Omega(\mathbf{z})=R(\mathbf{L}\mathbf{z})$,

    \item[{\makebox[1.1em][c]{\textbf{(b)}}}]
    magnitude: $\Omega(\mathbf{z})=\Omega_{\mathrm{Perspective}}^{(\mathbf{L},R,\gamma,\alpha)}(\mathbf{z})$ in Eq.~\eqref{eq:perspective_regularization},

    \item[{\makebox[1.1em][c]{\textbf{(c)}}}]
    log-magnitude: $\Omega(\mathbf{z})=\Omega_{\mathrm{EPILOG}}^{(\mathbf{L},R,\gamma,\alpha,\epsilon)}(\mathbf{z})$ in Eq.~\eqref{eq:EPILOG_regularization}.
    \vspace{2pt}
\end{itemize}
For each domain, we consider two regularization settings, labeled \textbf{(i)} and \textbf{(ii)}.
The setting \textbf{(i)} serves as a baseline without DCT.
Specifically, for \textbf{(a-i)}, we set $\mathbf{L}=\mathbf{I}$ and $R=\|\cdot\|_1$.
For \textbf{(b-i)} and \textbf{(c-i)}, we set $\mathbf{L}=\mathbf{I}$ and $R=0$ to exclude the term $\gamma R(\mathbf{L}\boldsymbol{\sigma})$.
In the setting \textbf{(ii)}, we set $\mathbf{L}$ to the DCT matrix and $R=\|\cdot\|_{\mathbf{w}}$ as defined in Eq.~\eqref{eq:weighted}.
The setting \textbf{(c-ii)} corresponds to the cepstral-domain sparsity proposed in Sect.~\ref{sec:EPILOG_cepstrum}.
All the algorithms were derived using the variable splitting scheme, detailed in Sect.~\ref{sec:EPILOG_admm}.
The number of iterations was set to 1000.
The methods were evaluated on 10 random pairs of speech signals and RIRs.

For each method, all the parameters were tuned to maximize SI-SNR \cite{rouxSDRHalfbakedWell2019}.
Five pairs of speech signals and RIRs were used for this tuning, which were distinct from those used for evaluation\footnote{
The optimized parameters $(\rho,\lambda,\alpha,\gamma,\epsilon)$ for each method were
\textbf{(a\text{-}i)}: $(0.95,0.35,\text{N/A},\text{N/A},\text{N/A})$,
\textbf{(a\text{-}ii)}: $(0.95,0.31,\text{N/A},\text{N/A},\text{N/A})$,
\textbf{(b\text{-}i)}: $(3.8\times10^{-3},0.78,0.44,2.93\times10^{-5},\text{N/A})$,
\textbf{(b\text{-}ii)}: $(0.67,0.66,$ $0.50,0.11,\text{N/A})$,
\textbf{(c\text{-}i)}: $(0.99,0.83,0.89,0.82,0.13)$, and
\textbf{(c\text{-}ii)}: $(0.99,$ $0.82,0.89,0.83,0.13)$, where N/A indicates an unused parameter.
}.

\subsection{Results}

First, we visually investigate the obtained signals.
Fig.~\ref{fig:placeholder} displays the estimated spectrogram $\mathbf{z}$ and its corresponding auxiliary variable $\boldsymbol{\sigma}$ obtained by each method, where $\mathbf{z}$ is presented in the log-magnitude domain $\log(|\mathbf{z}|)$ and the cepstral-domain $|\mathrm{DCT}(\log(|\mathbf{z}|))|$.
For methods \textbf{(b)}, which aim to regularize the magnitude, the auxiliary variable $\boldsymbol{\sigma}$ (middle) exhibited a typical structure of a magnitude spectrogram (see $|\mathbf{Gs}|$ in the rightmost-middle panel in Fig.~\ref{fig:placeholder} for comparison).
In contrast, for the proposed methods \textbf{(c)}, which aim to regularize the log-magnitude, $\boldsymbol{\sigma}$ (middle) closely resembled $\log(|\mathbf{z}|)$ (top), visually confirming the log-magnitude-link property of the EPILOG regularizer.
Moreover, the method \textbf{(c-ii)} yielded a sparser and smoother log-magnitude spectrogram $\log(|\mathbf{z}|)$ (top), resulting in fewer artifacts and improved performance compared to \mbox{\textbf{(c-i)}}.
This improvement is attributed to the cepstral-domain sparsity of $\mathbf{z}$ realized by regularizing the auxiliary variable $\boldsymbol{\sigma}$ (middle).
The cepstral representation $|\mathrm{DCT}(\log(|\mathbf{z}|))|$ (bottom) obtained by \mbox{\textbf{(c-ii)}} is sparser than that of \textbf{(c-i)}, illustrating the effect of cepstral-domain regularization via EPILOG in Sect.~\ref{sec:EPILOG_cepstrum}.

Next, we compare the methods quantitatively.
Table~\ref{tab:dereverberation_performance} summarizes the average dereverberation performance in terms of SI-SNR, STOI \cite{taal_short-time_2010}, and ViSQOL \cite{hines_visqol_2015}.
The method \textbf{(a-ii)} noticeably degraded performance compared to \textbf{(a-i)}, indicating that directly imposing DCT-domain sparsity on the complex-valued spectrogram was ineffective in this experiment.
The method \textbf{(b-ii)}, which promotes harmonic structure by imposing DCT-domain sparsity on the magnitude, showed competitive performance, though it did not reach the best score.
On the other hand, the proposed EPILOG-based methods \textbf{(c-i)} and \textbf{(c-ii)} outperformed the other methods.
The method \textbf{(c-i)}, where the regularization effect of $\gamma  R(\mathbf{L}\boldsymbol{\sigma})$ in Eq.~\eqref{eq:EPILOG_regularization} is excluded, achieved high performance, suggesting that the term $\Psi^{(\alpha,\epsilon)}
    (\mathbf{z},\boldsymbol{\sigma})$ itself serves as an effective sparsity-inducing regularizer.
Indeed, as shown in the top row of Fig.~\ref{fig:placeholder}, the log-magnitude spectrogram $\log(|\mathbf{z}|)$ obtained by \textbf{(c-i)} is sparser than those of \textbf{(a)} and \textbf{(b)}.
Furthermore, the method \textbf{(c-ii)}, where $\mathbf{L}$ corresponds to DCT, consistently achieved the highest performance across all metrics.
This result demonstrates the effectiveness of promoting cepstral-domain sparsity via the EPILOG regularizer.

\begin{table}[t]

\centering
\vspace{2.5pt}
\caption{
Average dereverberation performance.
The labels \textbf{(a)}--\textbf{(c)} indicate the regularization domain, and \textbf{(i)} and \textbf{(ii)} indicate regularization without and with the DCT, respectively.
Bold and underlined numbers represent the best and second-best values, respectively.
}
\vspace{-6pt}
\label{tab:dereverberation_performance}

\renewcommand{\arraystretch}{1.1}

\setlength{\tabcolsep}{1.5pt}

\begin{adjustbox}{max width=\columnwidth}

\begin{tabular}{c|c|cc|ccc}

\hline

&
Domain
& $\mathbf{L}$
& $\!\!\!R$
& SI-SNR $(\uparrow)$
& STOI $(\uparrow)$
& ViSQOL $(\uparrow)$

\\

\hline\hline

\rowcolor{gray!12}

\textbf{(a-i)}
&
& $\mathbf{I}$
& $\|\cdot\|_1$
& 15.260  & 0.911  & 3.314

\\

\rowcolor{gray!12}

\textbf{(a-ii)}
& \multirow{-2}{*}{Complex}
& $\mathrm{DCT}$
& $\|\cdot\|_{\mathbf{w}}$
& 11.440  & 0.883  & 3.024

\\

\hline

\textbf{(b-i)}
&
& $\mathbf{I}$
& $\!\!\!0$
& 15.250  & 0.910  & 3.310

\\

\textbf{(b-ii)}
& \multirow{-2}{*}{Mag.}
& $\mathrm{DCT}$
& $\|\cdot\|_{\mathbf{w}}$
& 15.216  & \underline{0.915}  & 3.382

\\

\hline

\rowcolor{gray!12}

\textbf{(c-i)}
&
& $\mathbf{I}$
& $\!\!\!0$
& \underline{15.670}  & 0.906  & \underline{3.437}

\\

\rowcolor{gray!12}

\textbf{(c-ii)}
& \multirow{-2}{*}{\shortstack{\vspace{-1pt}\\\textbf{Log-Mag.}\\\textbf{(ours)}}}
& $\mathrm{DCT}$
& $\|\cdot\|_{\mathbf{w}}$
& \textbf{15.871}  & \textbf{0.917}  & \textbf{3.497}

\\

\hline

\end{tabular}

\end{adjustbox}
\vspace{-5pt}
\end{table}

\section{CONCLUSION}

We proposed the EPILOG regularizer, a flexible technique for incorporating various prior knowledge in the log-magnitude domain. We established its log-magnitude-link property, derived its variable-wise proximity operators, and developed the proximal splitting algorithm tailored to EPILOG.
The speech dereverberation experiment demonstrated the effectiveness of cepstral-domain sparsity realized by EPILOG.
The proposed methodology opens the door to integrating a broader range of audio-specific priors in the log-magnitude domain into optimization-based signal processing.
Future work includes developing the algorithm with convergence guarantee and extending the application of EPILOG to a variety of signal processing tasks.

\newpage

\bibliographystyle{IEEEtran}
\bibliography{references}

\end{document}